\documentclass[10pt]{article}
\usepackage[utf8]{inputenc}
\usepackage{graphicx}
\usepackage{color}
\usepackage{graphicx}
\usepackage{url}
\usepackage{color,soul}
\usepackage{float}
\usepackage{amsmath}
\usepackage{amssymb}
\usepackage{courier}
\usepackage{multicol}
\usepackage[abs]{overpic}
\usepackage{enumitem}
\usepackage{multicol}
\usepackage{amsthm}
\usepackage{changepage}
\usepackage{framed}
\usepackage{mdframed}
\usepackage{lipsum}
\usepackage{longtable}

\usepackage[top=2.5cm, bottom=2.5cm, left=2.5cm, right=2.5cm]{geometry}
\usepackage{tikz}
\newtheorem{theorem}{Theorem}
\newtheorem{lemma}[theorem]{Lemma}
\newtheorem{claim}[theorem]{Claim}

\usepackage[linesnumbered,boxed]{algorithm2e}
\allowdisplaybreaks
\usepackage{pgf,tikz}
\usepackage{mathrsfs}
\usetikzlibrary{arrows}
\usepackage{times}
\DeclareMathSizes{10}{9}{6}{5}

\newcommand{\hide}[1]{}

\def\EE{\mathbb E}

\def\RR{\mathbb R}

\def\SS{\mathbb S}

\def\mI{\mathscr I}

\def\cH{\mathcal H}
\def\cI{\mathcal I}

\def\cS{\mathcal S}

\def\cR{\mathcal R}
\def\cL{\mathcal L}
\def\cM{\mathcal M}

\def\cH{\mathcal H}

\def\cV{\mathcal V}

\def\b1{\mathbf 1}

\def\cI{{\mathcal I}}

\def \OPT {\textsc{Opt}}
	
\def \RS {\textsc {RectStab}}
\def \HSS {\textsc {HorizSegStab}}	
\def \VSS {\textsc {VertSegStab}}
\def \SS {\textsc {SegStab}}
\def \USS {\textsc{UnitSqrStab}}
\def \stabs {\boxplus}

\newcommand{\raf}[1]{(\ref{#1})}

\title{Threshold Rounding for the Standard LP Relaxation of some Geometric Stabbing Problems}

\author{
	Khaled Elbassioni\thanks{Khalifa University of Science and Technology, SAN Campus, Abu Dhabi, UAE;
		(khaled.elbassioni@ku.ac.ae)}
	\and
	Saurabh Ray\thanks{New York University, Abu Dhabi, UAE;
	(saurabh.ray@nyu.edu) }
}
\date{}

\begin{document}

\maketitle

\begin{abstract}
    In the rectangle stabbing problem, we are given a set $\cR$ of axis-aligned rectangles in $\RR^2$, and the objective is to find a minimum-cardinality set of horizontal and/or vertical lines such that each rectangle is intersected by one of these lines. The standard LP relaxation for this problem is known to have an integrality gap of 2, while a better intergality gap of 1.58.. is known for the special case when $\cR$ is a set of horizontal segments. In this paper, we consider two more special cases: when $\cR$ is a set of horizontal and vertical segments, and when $\cR$ is a set of unit squares. We show that the integrality gap of the standard LP relaxation in both cases is stricly less than $2$. Our rounding technique is based on a generalization of the {\it threshold rounding} idea used by Kovaleva and Spieksma (SIAM J. Disc. Math 2006), which may prove useful for rounding the LP relaxations of other geometric covering problems.
\end{abstract}
\section{Introduction}
In this paper, we study two simple geometric covering problems both of which are special cases of the rectangle stabbing problem defined as follows: given a set $\mathcal{R}$ of axis-parallel rectangles in the plane, find the smallest set $\cS$ of horizontal or vertical lines so that each rectangle $r \in \mathcal{R}$ is intersected ({\em stabbed}) by at least one of the lines in $\cS$. 
This version of the problem is called the {\em continuous} version since the set $\cS$ can be any set of vertical and horizontal lines. In the {\em discrete} version of the problem, we are given a set of lines $\mathcal{L}$ and the solution set $\cS$ is required to be a subset of $\mathcal{L}$. Note that the discrete version is more general since in the continuous case we can always restrict to the set of lines that pass through one of the sides of one of the rectangles. The best approximation algorithm known for this problem is a $2$-approximation algorithm due to Gaur et. al.~\cite{GIK02} via LP-rounding. Ben-David et al.~\cite{BGMS12} 
show that there cannot be a better rounding of the LP-relaxation used in~\cite{GIK02} and hence the LP-relaxation has an integrality gap of $2$.
 Briefly, the idea is the following. We first solve the natural LP-relaxation in which we associate a value $x_\ell \in [0,1]$ with each line $\ell \in \mathcal{L}$ and minimize the sum of the values of the lines in $\mathcal{L}$ under the constraint that for each rectangle $r \in \mathcal{R}$, the total value of all lines intersecting $r$ is at least $1$. We then classify rectangles into two classes based on the solution to the LP. A rectangle $r$ is type-$h$ if the total value of the horizontal lines in $\mathcal{L}$ intersecting $r$ is at-least $0.5$ and type-$v$ otherwise. 
We solve two independent problems - one in which we stab the type-$h$ rectangles using the smallest number of horizontal lines in $\mathcal{L}$ and another in which we stab the the type-$v$ rectangles using the smallest number of vertical lines in $\mathcal{L}$. These two problems can be solved optimally since they are essentially problems of stabbing intervals with points, and union of these two solutions is easily shown to be of size at most twice the value of the LP-solution $\sum_{\ell \in \mathcal{L}} x_\ell$. The factor $2$ in the approximation algorithm here stems from the fact that the lines have two orientations - horizontal and vertical. A similar idea can be used to obtain a $2$-approximation for several other packing/covering problems including the ones we study in this paper. For many of these, the  $2$-approximation is the best that is known (in polynomial time), and in fact it is not known whether the natural LP has a smaller integrality gap.
Some simple examples of such problems are the following:
\begin{itemize}
    \item {\em Hitting segments with points.} Given a set of horizontal and vertical line segments in the plane, find the smallest number of points in the plane that {\em hit} (intersect) all the segments.
    \item {\em Covering points with segments.} Given a set of points in the plane and a set of horizontal or vertical segments whose union covers all the points, find the smallest subset of the segments that also cover all the points.
    \item {\em Independent set of segments.} Given a set of horizontal and vertical segments in the plane, find the largest subset of the segments that are pairwise non-intersecting. 
\end{itemize}
It is a natural to ask whether the integrality gap of the standard LP relaxation for these problems is smaller than $2$. An interesting case where an integrality gap smaller than $2$ was shown is  a special case of the rectangle stabbing problem where all the rectangles have height $0$, i.e., the rectangles are horizontal line segments. Kovaleva and Spieksma \cite{KS06} showed that the standard LP relaxation has an integrality gap of exactly $e/(e-1) \approx 1.582$. We give a simpler proof of the upper bound of $e/(e-1)$ and show that the technique can be extended to beat the obvious bound of $2$ on the integrality gap for the following extremal special cases~\footnote{one extremal special case which we were not able to prove an integrality gap smaller than $2$ is where all rectangles are squares but possibly of different sizes} of the rectangle stabbing problem: i) all rectangles have either width $0$ or height $0$, i.e., they are either horizontal or vertical segments and, ii) all rectangles are unit-sized squares. We show an integrality gap of at most $1.935$ for the first problem and at most $1.98\overline{3}$~\footnote{In the appendix, we outline a proof showing that this bound can be improved to $1.9\overline{3}$. We believe that further improvement may be possible. However, the main focus in this paper is simply to prove a bound that is strictly less than $2$.} for the second problem. Our proofs are constructive and yield simple deterministic algorithms to round the LP-solution with the same approximation ratio. These are also currently the best polynomial time approximation algorithms known for these problems. 
Numerically, these are not great improvements but technically it seems non-trivial to prove any bound below $2$. We believe that our approach could lead to better techniques to round LPs for basic geometric packing and covering problems. 
For hardness results and other variants related to the rectangle stabbing problem, see~\cite{CDSW18,DFRS12,ELRSSS08,XX05}.
\subsection{Related Work}
Packing and covering problems are very well studied in computational geometry. Broadly, there are three main algorithmic techniques: LP-rounding, local search and separator based methods.

The technique of Bronimann and Goodrich~\cite{BG95} reduces any covering problem to an $\epsilon$-net question so that if the set system admits an $\epsilon$-net of size $\frac{1}{\epsilon} \cdot f(\frac{1}{\epsilon})$, then we obtain an LP-relative approximation algorithm with approximation factor $f(\OPT)$ where $\OPT$ is the size of the optimal solution. Since set systems of finite VC- dimension admit $\epsilon$-nets of size $O(\frac{1}{\epsilon} \log \frac{1}{\epsilon})$, this implies an $O(\log \OPT)$ approximation algorithm for covering problems involving such set systems. Similarly, for set systems with low shallow cell complexity, we obtain algorithms with correspondingly small approximation factors (see~\cite{V10, chan2012weighted}). 
In particular if the shallow cell complexity is linear, we obtain constant factor approximation algorithms.  Varadarajan~\cite{V10} showed via the \emph{quasi-uniform} sampling technique how these results can be made to work in the weighted setting. His technique was optimized by Chan et. al.~\cite{chan2012weighted} who also introduced the notion of \emph{shallow cell complexity} generalizing the notion of {\em union complexity} from geometric set systems to abstract set systems. 
Some of these algorithms have also been extended to work in the multicover setting (see \cite{ChekuriCH12}, \cite{BansalP16}). 
One limitation of the approach in \cite{BG95} is that even for simple set systems with linear shallow cell complexity, the lower bound on the size of the $\epsilon$-net may involve a large constant factor which then translates into a lower bound on the  approximation ratio of the corresponding rounding algorithm. For the independent set problem Chan and Har-Peled showed that LP-rounding can be used to obtain a constant factor approximation for pseudodisks. For rectangles in the plane, Chalermsook and Chuzhoy gave an $O(\log \log n)$ approximation algorithm. Recently, Chalermsook and Walczak~\cite{CW21} extended this to work for the weighted setting as well. Chan~\cite{Chan12} obtained an $O(n^{0.368})$-approximation algorithm for the {\em discrete} independent set problem with axis parallel rectangles. 

The local search framework, where one starts with any feasible solution and tries to improve the solution by only making constant size {\em swaps} (i.e., adding/removing a constant number of elements from the solution), yields a PTAS for several packing and covering problems (see e.g.~\cite{MR10, ChanH12, Aschner2013, GibsonP10, RR18, BasuRoy2018}). This framework also has major limitations: it is not as broadly applicable as the LP-rounding technique (in particular it works only for the unweighted setting so far), often hard to analyse, and the PTASes they yield have a  running time like $n^{O({1/\epsilon^2})}$ with large constants in the exponent, making them  irrelevant for practical applications. 

The third type of algorithms consists of separator based methods where some kind of {\em separator} is used to split the problem instance into smaller problems which can be solved independently and combined to obtain an approximate solution. Hochbaum and Maass~\cite{HM87} used this idea to obtain PTASes for several packing and covering problems. More recently, Adamaszek and Wiese~\cite{AW13, AW14} have used this kind of idea for obtaining a QPTAS for independent set problems. Mustafa et. al.~\cite{mustafa2015quasi} extend the idea of Adamaszek and Wiese to obtain a QPTAS for weighted set cover problem with pseudodisks in the plane and halfspaces in $\mathbb{R}^3$. 

Unfortunately none of these broad techniques yield better than a $2$-approximation algorithm (in polynomial time) even for the special cases of the rectangle stabbing problem that we consider. There are simple examples showing that the solution returned by the standard local search algorithm can be arbitrarily bad. Moreover, it is not clear whether any of the separator based techniques can be adapted to work for this problem.

\hide{
Several geometric problems, mostly of the covering type, have the property that their standard LP relaxation can be rounded within a factor of 2 with a simple algorithm, but no matching lower bound is known on the intergality gap. We mention a few examples: covering points in the plane with horizontal/vertical line segments, and the dual problem;... Understanding the standard LP relaxation for such problems seems to be natural and well-motivated question.
}

\medskip
\subsection{Problem definition and preliminaries}
In this paper, we consider special cases of the following problem:

\begin{description}
\item \RS$(\cR,\cH,\cV,w)$: given a set $\cR$ of axis-aligned rectangles in $\RR^2$, a set $\cH$ of horizontal lines, a set $\cV$ of vertical lines, and non-negative weights $w:\cH \cup \cV \to\RR_+$, find a minimum-weight collection of lines $\cL\subseteq\cH\cup\cV$ such that each rectangle in $\cR$ is stab (intersected) by at least one line from $\cL$.
\end{description}

In the {\it unweighted} version, $w_\ell=1$ for all $\ell\in\cH\cup\cV$, and in the {\it continuous} version, $\cH$ and $\cV$ are the sets of {\it all} horizontal and vertical lines in the plane, respectively. It is natural to assume that the continuous version is unweighted. By considering the two sets of intervals obtained by projecting the rectangles along the horizontal (resp., vertical) axis, and adding a vertical (resp., horizontal) line for each maximal clique (that is, maximal set of intersecting intervals), we may assume further that the given sets of lines are {\it discrete}.

The special cases of \RS\ in which $\cR$ is a set of horizontal (resp., vertical) line segments, axis-aligned line segments, and unit squares will be denoted by \HSS\ (resp.,\VSS), \SS, and \USS, respectively. 

In the standard LP relaxation for \RS, we assign a variable $x_i$ to each vertical line $i\in\cV$, a variable $y_i$ to each horizontal line $j\in\cH$, and demand that, for each rectangle $r\in\cR$, the sum of the variables corresponding to the lines $i\in\cH\cup\cV$ stabbing $r$ (denoted henceforth as $i\stabs r$) is at least 1:

\begin{align}\label{LP-RS}
z^*:=&\min\quad \sum_{i\in\cV}w_ix_i+\sum_{j\in\cH}w_jy_j\\
\text{s. t.} &\sum_{i\in\cV:~i\stabs r}x_i+\sum_{j\in\cH:~j\stabs r}y_j\ge 1,\quad \forall r\in \cR\label{LP-RS-1}\\
&x_i\ge 0,\quad \forall i\in\cV, \qquad y_j\ge 0,\quad \forall j\in\cH.\label{LP-RS-2}
\end{align}

Note that if one of $\cH$ or $\cV$ is empty, the above LP relaxation is {\it exact} due to {\it total unimodularity} of the (interval) constraint matrix.
Gaur et al. \cite{GIK02} showed that the integrality gap of the relaxation~\raf{LP-RS}-\raf{LP-RS-2} is at most 2 via the following simple rounding algorithm. Given an instance \RS$(\cR,\cH,\cV,w)$, let $(x^*,y^*)$ be an optimal solution to the relaxation. Define two new instances: $\cI_1:=$\RS$(\{r\in\cR:~\sum_{i\in\cV:~i\stabs r}x_i^*\ge\frac12\},\emptyset,\cV,w)$ and $\cI_2:=$\RS$(\{r\in\cR:~\sum_{j\in\cH:~j\stabs r}y_j^*\ge\frac12\},\cH,\emptyset,w)$. Since the two instances $\cI_1$ and $\cI_2$ have an integrality gap of 1, they can be efficiently solved to optimality to yield two sets of lines $\cH'\subseteq\cH$ and $\cV'\subseteq\cV$ whose union is a 2-approximation for the original instance of \RS. Ben-David et al. \cite{BGMS12} proved that this is essentially tight.

For the special case of \HSS\ (or \VSS), Kovaleva and Spieksma \cite{KS06} showed that the integrality gap of the relxation~\raf{LP-RS}-\raf{LP-RS-2} is exactly $\frac{e}{e-1}=1.58..$. They achieve the upper bound for any instance \HSS\ $(\cR,\cH,\cV,w)$ as follows. Given an optimal solution $(x^*,y^*)$ to the LP-relaxation, consider the components of $y^*$ in a non-decreasing order, say w.l.o.g.,  $y_1^*\le y_2^*\le\cdots\le y_m^*$, where $m:=|\cH|$. Define $m+1$ new instances $\cI_k:=$\RS$(\cR_k,\emptyset,\cV,w)$, for $k=1,\ldots,m+1$, where $\cR_k:=\{r\in\cR:~j\stabs r\text{ for some }j< k\}$. The $k$th such instance corresponds to rounding $y_j$ to $1$ if and only if $y_j^*\ge y_k^*$, where we assume $y_{m+1}^*=1$. Since $\frac{x^*}{1-y_k^*}$ is a feasible solution for the relaxation of $\cI_k$ (which has an integrality gap of 1), it follows that the total weight of the rounded solution of the given instance \HSS$(\cR,\cH,\cV,w)$ is upper-bounded by
\begin{align}\label{KS-UB}
    \min_{k=1,\ldots,m+1}\left(\sum_{j=k}^mw_j+\frac{w(x^*)}{1-y_k^*}\right),
\end{align}
where, for brevity, we write $w(x):=\sum_{i\in\cV}w_ix_i$ and $w(y):=\sum_{j\in\cH}w_jy_j$.
Kovaleva and Spieksma prove (by induction) that the expression in~\raf{KS-UB} is bounded from above by $\frac{e}{e-1}\left(w(x^*)+w(y^*)\right)$. 

\paragraph{Threshold rounding.} A slightly simpler (and probably more intuitive) way to see the result in \cite{KS06} is the following.
Fix a constant $\rho\in(0,1)$ (to be determined later). Given an optimal solution $(x^*,y^*)$ to the LP-relaxation of \HSS$(\cR,\cH,\cV,w)$, choose a {\it threshold} $\tau$ randomly in $[0,\rho]$. For $j=1,\ldots,m$, round $y_j$ to 1 if and only if $y_j^*\ge\tau$. Let $\widehat y$ be the rounded solution, $k$ be the smallest index such that $y_k^*\ge\tau$ (where we assume as before that $y^*_{m+1}=1$), and $\widehat x$ be an optimal integral solution of the instance $\cI_k$. Then $(\widehat x,\widehat y)$ is a feasible solution for \HSS$(\cR,\cH,\cV,w)$ with objective
\begin{align}\label{HS-ThR}
w(\widehat x)+w(\widehat y)\le\sum_{j:~y_j^*\ge\tau}w_j+\frac{w(x^*)}{1-\tau},
\end{align}
since $\frac{x^*}{1-\tau}$ is feasible for the relaxation of $\cI_k$.
Taking the expectation over all possible values of $\tau$,
\begin{align}\label{HS-ThR-Exp}
   \EE[w(\widehat x)+w(\widehat y)]&\le
   \sum_jw_j\Pr[y_j^*\ge\tau]+
   \frac{1}{\rho}\int_0^{\rho}\frac{w(x^*)}{1-\tau}d\tau
   =\frac{1}{\rho}w(y^*)+\frac{1}{\rho}\ln\left(\frac{1}{1-\rho}\right)w(x^*).
\end{align}
By choosing $\rho=1-\frac1e$, to balance the coefficients of the two terms in \raf{HS-ThR-Exp}, we get $\EE[w(\widehat x)+w(\widehat y)]\le\frac{e}{e-1}\left(w(x^*)+w(y^*)\right)$.

Kovaleva and Spieksma~\cite{KS06} also give tight examples where the integrality gap can be made arbitrarily close to $e/(e-1)$.\\

In this paper we extend the threshold rounding technique to show an integrality gap of less than $2$ for the standard LP relaxations 
of \SS and \USS.
We need to recall the following definitions. 

\medskip

\noindent{\bf $(\ell,\delta)$-Shifted partition.}
Given a line segment $L$ of length $|L|$ and a parameter $\ell>0$, consider a partition of $L$ into $k=\lceil|L|/\ell\rceil$ consecutive segments $L_1,\ldots,L_k$, each, except possibly the last, having length $\ell$.  For $\delta>0$, we shift the start of each interval $L_i$ by the same amount $\delta$, wrapping around $L_k$ such that the resulting intervals still partition $L$.  We call such a  configuration of intervals an $(\ell,\delta)$-shifted partition of $L$. Given any sub-interval $I$ of $L$, we say that $I$ is {\it crossed} by the partition if $I$ contains an end-point of one of the segments $L_1, \cdots, L_k$ in the partition in its interior.

\medskip

\noindent{\bf $\epsilon$-net.} Let $V$ be a given set of points on the line with non-negative ``values" $x_v$ and ``weights" $w_v$, for $v\in V$, and $\cI$ be a set of intervals over $V$. For any parameter $\epsilon>0$, an $\epsilon$-net (w.r.t. $(V,\cI,x)$) is a subset $S\subseteq V$ such that $S\cap I\ne\emptyset$ for every interval $I\in\cI$ with $x(I)\ge \epsilon$, where $x(I):=\sum_{v\in I}x_v$. It is well-known that there always exists an $\epsilon$-net $S$ with total weight $\sum_{v\in S}w_v \le\sum_{v\in V}w_vx_v/\epsilon$. Indeed, assume that the points are given in the order $v_1,\ldots,v_n$ on the line, and let us consider the line segment $L$ obtained by stacking together consecutive intervals $I_{v_1},\ldots,I_{v_n}$ of respective lengths $x_{v_1},\ldots,x_{v_n}$. Choose $\delta\in[0,\epsilon]$ randomly, consider an $(\epsilon,\delta)$-partition of $L$, and define  $S:=\{v\in V:~I_v \text{ is crossed by the partition}\}$. Clearly, $S$ is (almost surely) an  $\epsilon$-net, and moreover, $\Pr[v\in S]=x_v/\epsilon$, implying that $\EE[\sum_{v\in S}w_v]=\sum_{v\in V}w_vx_v/\epsilon$.
Note that an optimal (i.e., with smallest total weighted value ) $\epsilon$-net can be obtained (in deterministic polynomial time) by solving the instance \RS$(\{I\in\cI: x(I)\ge\epsilon\},\cV,\emptyset,w\})$, where $\cV$ is the set of vertical lines through the points in $V$.

\section{Segment Stabbing}
We now show that $\SS$ has an integrality gap of less than 2.

{\bf Algorithm.} We fix two thresholds $\tau_x$ and $\tau_y$ in $(0,1)$. Let $(x^*, y^*)$ be the optimal solution to the LP relaxation. We round any $x^*_i \geq \tau_x$ to $1$ and similarly we round any $y^*_j \geq \tau_y$ to $1$. Now consider any horizontal segment which has not yet been hit by the segments chosen due to this rounding. This means that the LP value of the horizontal line containing the segment is  less than $\tau_y$, implying that total LP value of all the vertical lines hitting this segment is at least $1-\tau_y$. Similarly, for any vertical segment not yet hit by the segments chosen due to rounding, the total LP value of all the horizontal lines hitting the segment is at least $1-\tau_x$. Let $X^*$ denote the contribution to the LP objective by the vertical segments not yet rounded and similarly let $Y^*$ denote the contribution of the horizontal segments not yet rounded to the LP objective. We hit the remaining horizontal segments with a {\em vertical} $(1-\tau_y)$-net of weight at most $X^*/(1-\tau_y)$. Similarly, we hit the remaining vertical segments with a {\em horizontal} $(1-\tau_x)$-net of weight at most $Y^*/(1-\tau_x)$. 

It remains to specify how the values $\tau_x$ and $\tau_y$ are chosen. We first pick a $\tau$ from a distribution\footnote{The reader may find the distribution on $\tau$ a bit ad-hoc. In fact, we have obtained this by trial and error. More complicated distributions yield a slightly better bound but for clarity we chose the simplest distribution that suffices.} defined in the interval $[\alpha, \beta]$ where $\alpha = 0.25$ and $\beta = 0.45$ as follows. The density at any $\tau \in [\alpha, \beta]$ is $\rho(\tau) = \frac{2(\tau-\alpha)}{(\beta-\alpha)^2}$. We then set $\tau' = h(\tau) := 1 - (1-\beta)^2/(1-\tau)$. We note that for any  $\tau \in [\alpha, \beta]$, $\tau' \geq \tau$ and as $\tau$ increases from $\alpha$ to $\beta$, $\tau'$ decreases from  $\gamma := h(\alpha) = 0.59\overline{6}$ to $\beta$.
Finally, we either set ($\tau_x = \tau$ and $\tau_y = \tau'$) or ($\tau_x = \tau'$ and $\tau_y = \tau$) with equal probability.  

{\bf Analysis.}
For any fixed $\tau_x$ and $\tau_y$,
the total weight of the solution obtained by the algorithm described above is at most 
\begin{align}
 &   \sum_{i \in \cV: x^*_i \ge \tau_x} w_i 
+ \frac{X^*}{1-\tau_y} 
+ \sum_{j \in \cH: y^*_j \ge \tau_y} w_j  + \frac{Y^*}{1-\tau_x} \\  
=&
\sum_{i \in \cV: x^*_i \ge \tau_x} w_i 
+ \frac{1}{1-\tau_y} \sum_{i \in \cV: x^*_i < \tau_x} w_ix^*_i
+ \sum_{j \in \cH: y^*_j \ge \tau_y} w_j  + \frac{1}{1-\tau_x} \sum_{j \in \cH: y^*_j < \tau_y} w_jy^*_j .
\end{align}

Consider any vertical line $\ell$ with LP value $x^*_\ell$. If $x^*_\ell \ge \tau_x$, it contributes $w_\ell$ to the solution in the first term of the above sum.  Otherwise, it contributes $\frac{w_\ell x_\ell^*}{1-\tau_y}$ in the second term of the sum. We can thus write the contribution of the line $\ell$ as $\lambda_\ell \cdot w_\ell x^*_\ell$ where $\lambda_\ell = 1/{x^*_\ell}$ if $x^*_\ell \geq \tau_x$ and $1/(1-\tau_y)$ otherwise. 

For any fixed choice of $\tau \in [\alpha, \beta]$, the expected value of $\lambda_\ell$ is $\mu(x^*_\ell, \tau)$ where $\mu(z,\tau)$ is defined as

\begin{align}
     {\mu(z,\tau)} =\left\{\begin{array}{ll}
    \frac1{2}\left(\frac1{1-\tau}+\frac1{1-\tau'}\right) &\text{ if } z <\tau\\
     \frac1{2}\left(\frac1{z} + \frac1{1-\tau}\right) &\text{ if } \tau \leq z <\tau'\\
     \frac1{z}  &\text{ if } z >\tau',
    \end{array}
    \right.
\end{align}
where we recall that $\tau' = h(\tau)$.
The expected value of $\lambda_\ell$ over the choices of $\tau$ is therefore $\overline\mu(x^*_l)$ where $\overline\mu(z) = \int_{\alpha}^{\beta} \mu(z,\tau)\,\rho(\tau)\,d\tau$.
We show in the lemma below that for any $z \in [0,1]$, $\overline\mu(z) < 1.935$. This implies that in expectation, the contribution of any vertical line $\ell$ to the solution is at most $1.935 \cdot w_\ell x_\ell^*$. An analogous proof shows that an analogous claim holds for any horizontal line $\ell$.

\begin{lemma}
For any $z \in [0,1]$, $\overline\mu(z) < 1.935$. 
\end{lemma}
\begin{proof}
We do a case analysis based on the value of $z$. The cases below cover all possible values of $z$ in $[0,1]$ but overlap on the boundaries.
\begin{itemize}
    \item {\bf Case 1}. $z \in  [0,\alpha]$. In this case for any choice of $\tau$, $z < \tau$ and therefore $\mu(z,t) = \frac{1}{2}\left( \frac1{1-\tau} + \frac1{1-\tau'}\right)$. Thus 
    $$ \overline\mu(z) = \int_\alpha^\beta \mu(z,\tau)\,\rho(\tau)\,d\tau
    = \int_{\alpha}^\beta \left(\frac1{1-\tau} + \frac{1-\tau}{(1-\beta)^2} \right) \cdot \frac{\tau-\alpha}{(\beta- \alpha)^2}\,d\tau
   = \int_{\alpha}^\beta \left( A_1\tau^2 + B_1\tau + C_1 + \frac{D_1}{1-\tau} \right) d\tau,
    $$
    where $A_1 = -\frac{1}{(1-\beta)^2(\beta-\alpha)^2},  B_1 = \frac{1+\alpha}{(1-\beta)^2\cdot (\beta-\alpha)^2}, C_1 = -\frac{1}{(\beta-\alpha)^2} - \frac{\alpha}{(1-\beta)^2(\beta-\alpha)^2}$ and $ D_1 = \frac{1-\alpha}{(\beta-\alpha)^2}$.
    Thus,
    $$\overline \mu(z) = \left[A_1\frac{\tau^3}{3} + B_1\frac{\tau^2}{2} + C_1\tau - D_1 \ln (1-\tau) \right]_\alpha^\beta \approx 1.835. $$
    
    \item {\bf Case 2.} $z \in [\alpha, \beta]$. In this case, 
    \begin{align*}
        \overline\mu(z) = \int_\alpha^\beta \mu(z,\tau)\,\rho(\tau)\,d\tau 
        = \frac12\int_\alpha^z \left(\frac{1}{z} + \frac{1}{1-\tau}\right) \rho(\tau)\,d\tau
\ +\ \frac12 \int_z^\beta \left( \frac{1}{1-\tau} + \frac{1-\tau}{(1-\beta)^2} \right)  \rho(\tau)\,d\tau
    \end{align*}
    The derivative of the above expression with respect to $z$ (using Leibniz integral rule) is
    $$\frac12\left(\frac{1}{z} + \frac{1}{1-z}\right) \rho(z) -\frac12\int_\alpha^z\frac{\rho(\tau)}{z^2}\,d\tau - \frac12\left( \frac{1}{1-z} + \frac{1-z}{(1-\beta)^2} \right) \rho(z) 
 = \frac12\left(\frac{1}{z} - \frac{1-z}{(1-\beta)^2}-\frac{z-\alpha}{2z^2} \right) \rho(z). $$
    Note that $f(z):=\frac{1}{z} - \frac{1-z}{(1-\beta)^2}-\frac{z-\alpha}{2z^2}$ has derivative $f'(z)= \frac{1}{(1-\beta)^2}-\frac{1}{2z^2}-\frac{\alpha}{z^3}$, which is non-decreasing in $z$, implying that $f'(z)\le f'(\beta)=\frac{1}{(1-\beta)^2}-\frac{2\alpha+\beta}{2\beta^3}\approx-1.907$, for $z\in[\alpha,\beta]$. Hence, $f(z)$ is decreasing with $f(\alpha)=\frac{1}{\alpha}-\frac{1-\alpha}{(1-\beta)^2}\approx1.521$ and $f(\beta)=\frac{\alpha+\beta}{2\beta^2}-\frac{1}{1-\beta}\approx-0.09$, which in turn implies that $\overline\mu(z)$ has a unique maximum in $[\alpha,\beta]$ at the point $z_0\in[\alpha,\beta]$ satisfying $f(z_0)=0$. It can be be verified that $z_0\approx0.414$. Thus, $\overline\mu(z) \leq \overline\mu(z_0)$ for $z\in[\alpha,\beta]$. Now, 
    \begin{align*} 
    \overline\mu(z_0) & = \frac12\int_\alpha^{z_0} \left(\frac{1}{z_0} + \frac{1}{1-\tau}\right) \rho(\tau)\,d\tau +\ \frac12 \int_{z_0}^\beta \left( \frac{1}{1-\tau} + \frac{1-\tau}{(1-\beta)^2} \right)  \rho(\tau)\,d\tau\\
    &= \int_\alpha^{z_0} \left(\frac{1}{z_0} + \frac{1}{1-\tau}\right) \cdot \frac{\tau-\alpha}{(\beta-\alpha)^2}\,d\tau +\  \int_{z_0}^\beta \left(\frac1{1-\tau} + \frac{1-\tau}{(1-\beta)^2} \right) \cdot \frac{\tau-\alpha}{(\beta- \alpha)^2}\,d\tau\\
    &=\int_\alpha^{z_0} \left(A_2\tau + B_2 + \frac{C_2}{1-\tau}\right) d\tau +\ \int_{z_0}^\beta \left( A_1\tau^2 + B_1\tau + C_1 + \frac{D_1}{1-\tau} \right) d\tau,
    \end{align*}
    where $A_2 =\frac1{z_0\cdot (\beta-\alpha)^2}, B_2 = - \frac{\alpha}{z_0 \cdot (\beta-\alpha)^2} - \frac{1}{(\beta-\alpha)^2}$, $C_2 =\frac{1-\alpha}{(\beta-\alpha)^2}$, and $A_1$, $B_1$, $C_1$, $D_1$ are as above.
    Thus $$ \overline\mu(z_0) = \left[A_2\frac{\tau^2}{2} + B_2\tau - C_2\ln(1-\tau)\right]_\alpha^{z_0}+\left[A_1\frac{\tau^3}{3} + B_1\frac{\tau^2}{2} + C_1\tau - D_1 \ln (1-\tau) \right]_{z_0}^\beta\approx 1.9347.$$ This implies that $\overline\mu(z) < 1.935$ for all $z \in [\alpha, \beta]$.
    
    \item {\bf Case 3.} $z \in [\beta, \gamma]$. In this case, 
   $$ \overline\mu(z) = \int_\alpha^\beta \mu(z,\tau)\,\rho(\tau)\,d\tau = 
\frac12\int_\alpha^{h(z)} \left(\frac{1}{z} + \frac{1}{1-\tau}\right) \rho(\tau)\,d\tau
\ +\ \frac12\int_{h(z)}^\beta \frac{2}{z}\, \rho(\tau)\,d\tau.
$$
Let $w = h(z) \in [\alpha, \beta]$. In terms of $w$, the above expression is 
$$ 
\int_\alpha^{w} \left(\frac{1}{h(w)} + \frac{1}{1-\tau}\right) \rho(\tau)\,d\tau
\ +\ \int_{w}^\beta \frac{2}{h(w)}  \rho(\tau)\,d\tau
$$
The derivative of the above expression with respect to $w$ is 
\begin{align} \label{e1}
&\left(\frac{1}{h(w)} + \frac{1}{1-w}\right)\rho(w)-\frac{h'(w)}{h^2(w)}\int_\alpha^w\rho(t)\, d\tau-\frac{2}{h(w)}\rho(w)-\frac{2h'(w)}{h^2(w)}\int_w^\beta\rho(t)\, d\tau\\ 
&=\left[ \frac{1}{1-w} - \frac{1}{h(w)}+\frac{h'(w)}{2h^2(w)}(w-\alpha)\right] \rho(w)-\frac{2h'(w)}{h^2(w)},\label{e2}
\end{align}
where $h'$ denotes the derivative of $h$. It can be verified that the right-hand side $g(w)$ of \raf{e1}-\raf{e2} is decreasing in $w\in[\alpha,\beta]$, and hence has value at least $g(\beta)\approx0.898$. This implies that the integral is maximized when $w = \beta$, i.e., when $z = \beta$. Similar to  the previous case, we can estimate the integral by
$$\overline\mu(\beta)=\left[A_3\frac{\tau^2}{2} + B_3\tau - C_3\ln(1-\tau)\right]_\alpha^{\beta}
\approx 1.927,$$
where $A_3 =\frac1{\beta\cdot (\beta-\alpha)^2}, B_3 = - \frac{\alpha}{\beta} - \frac{1}{(\beta-\alpha)^2}$ and $C_3 =\frac{1-\alpha}{(\beta-\alpha)^2}$. This implies that $\mu(z) < 1.927$  for all $ z \in [\beta, \gamma]$. 
    
\item {\bf Case 4.} $z \in [\gamma, 1]$.  In this case, 
$$ \overline\mu(z) = \int_\alpha^\beta \mu(z,\tau)\,\rho(\tau)\,d\tau = \int_{\alpha}^{\beta} \frac{1}{z} \cdot \rho(\tau) \,d\tau$$
which is clearly maximized when $z = \gamma$ and $\overline\mu(\gamma) = 1/\gamma \approx 1.68$. Thus, $\overline\mu(z) \leq 1.7$ for all $z \in [\gamma,1]$.
\end{itemize}
The lemma follows.
\end{proof}
The theorem below follows from the above discussion.
\begin{theorem}\label{t1}
The integrality gap of $\SS$ is at most $1.935$.
\end{theorem}

{\bf Remark.} Even though we pick $\tau_x$ and $\tau_y$ randomly in the algorithm above, it can be easily converted to a deterministic algorithm since we only need to consider the thresholds $\tau_x$ and $\tau_y$ among the values of the variables in the optimal LP solution. A similar remark applies to the algorithm given in the next section. 

\paragraph{Limitation of the analysis.} We now show that our analysis of the algorithm described above cannot be improved significantly. We do this by showing that there exist $x^*$ and $y^*$ s.t. no matter what $\tau_x$ and $\tau_y$ are chosen the weight of the solution output by the algorithm is at least $1.89\,(\sum_{i}x^*_i + \sum_j y^*_j )$. We will assume that each line (vertical or horizontal) has weight $1$. Our choice of $x^*$ and $y^*$ will in fact identical, i.e., they will have the same set of values. We will thus only state what the values in $x^*$ are. 

The values in $x^*$ are in $(0, 0.5)$ and their distribution is chosen according to a density  function $f$  defined below (see Figure~\ref{fig:f0}). In other words, the number of values in $x^*$ that lie between $a$ and $a+\delta$ for any $a \in (0,0.5-\delta)$ is proportional to $f(a) \cdot \delta$ for some sufficiently small $\delta$. The function $f$ is defined in $[0,1]$ as follows:
\begin{align*}
     f(t) =\left\{\begin{array}{ll}
     0.5 &\text{ if } x \in (0, 0.2)\\
     18.75x - 3.25 &\text{ if } x \in [0.2, 0.4)\\
     4.25  &\text{ if } x \in [0.4, 0.5)\\
     0 &\text{ otherwise } 
    \end{array}
    \right.
\end{align*}

\begin{figure}
\begin{center}
    \includegraphics[scale=0.75]{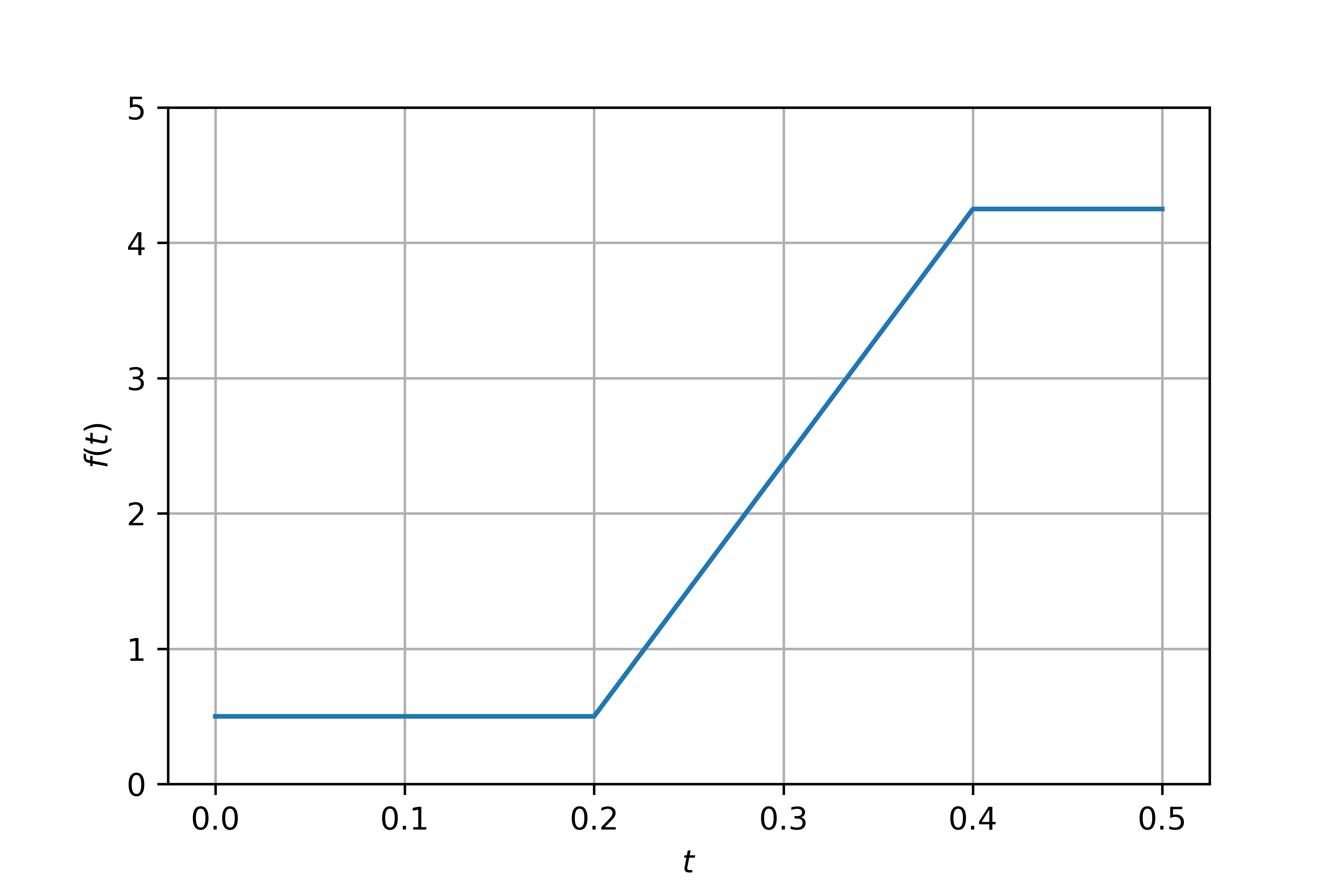}
\end{center}
\caption{The density function $f$ used in the tightness example.}\label{fig:f0}
\end{figure}

More specifically, $x^*$ consists of $\lfloor N\cdot f(\frac{i}{N}) \rfloor$ numbers uniformly distributed in the interval $(\frac{i}{N}, \frac{i+1}{N})$ for each $i \in \{0, \cdots, N-1\}$, where $N$ is a sufficiently large integer.  
For any choice of $\tau_x, \tau_y \in (0,1)$, the 
weight of the solution returned by the algorithm is 

\begin{align*}
W &= |\{i: x^*_i > \tau_x\}| + |\{j: x^*_j > \tau_y\}|  + \frac{1}{1-\tau_x} \sum_{j:x^*_j \le \tau_y} x^*_j + \frac{1}{1-\tau_y}\sum_{i:x^*_i \le \tau_x} x^*_i\\
&= N^2 \cdot \left[ \int_{\tau_x}^1 f(x)\,dx
\ +\  \int_{\tau_y}^1 f(t)\,dt
 +  \frac{1}{1-\tau_x} \int_0^{\tau_y} tf(t)\,dt + \frac{1}{1-\tau_y}\int_0^{\tau_x} tf(t)\,dt \right]\ \pm O(N).
\end{align*}
Recall that $y^*$ is the same as $x^*$. 
Since the ``LP-solution" $(x^*,y^*)$ has value $V = 2\sum_{i} x^*_i = 2N^2 \int_{0}^1 t f(t)\, dt \  \pm O(N)$, the ratio $W/V$
 can be made arbitrarily close to 
$$
\gamma(\tau_x, \tau_y) = 
\frac{1}{2\int_0^1 tf(t)\,dt} \cdot \left[ \int_{\tau_x}^1 f(t)\,dt
\ +\  \int_{\tau_y}^1 f(t)\,dt
 +  \frac{1}{1-\tau_x} \int_0^{\tau_y} tf(t)\,dx + \frac{1}{1-\tau_y}\int_0^{\tau_x} tf(t)\,dt \right]
$$ 
 by choosing a sufficiently large $N$. It can be shown that $\gamma(\tau_x, \tau_y) \geq 1.89$ for any $\tau_x, \tau_y \in [0,1]$. We skip the technical proof since it does not yield much insight and since this can be easily checked numerically.\\
 
 \noindent
 {\bf Remark.} What we have shown above is that our method of analysis cannot yield a bound better than $1.89$.  Note however that this is not necessarily a lower bound on the integrality gap of the LP-relaxation or even on the approximation factor of our algorithm since the $x^*$ we used above need not correspond to the solution of the LP-relaxation for a valid instance of the problem. The currently best lower bound remains the one of $e/(e-1)$ shown in~\cite{KS06}.

\section{Unweighted Continuous Unit Square Stabbing}
It is natural to try to apply the threshold rounding idea to \USS. Given an optimal solution $(x^*,y^*)$ for the LP relaxation~\raf{LP-RS}-\raf{LP-RS-2}, where $\cR$ is a given set of unit squares and $w\equiv1$, we pick two thresholds $\tau_x$ and $\tau_y$ according to some distribution, such that $\tau_x+\tau_y=1$. Let $\cI_x:=\{r_x:~r\in\cR\}$ and $\cI_y:=\{r_y:~r\in\cR\}$ be the two sets of projections of the squares on the horizontal and vertical axes, respectively. We construct a $\tau_x$-net for the set of segments in $\cI_x$ and a $\tau_y$-net for those in $\cI_y$. Since $\tau_x+\tau_y=1$, for any $r\in\cR$, we must have either $x^*(r_x):=\sum_{i\in\cV:~i\stabs r}x_i^*\ge \tau_x$ or  $x^*(r_y)\ge \tau_y$, and thus, the union of the two constructed nets is a feasible solution to the given instance of \USS. It remains to argue about the quality of this solution. If one would use the simple bounds $\sum_{i\in\cV}x_i^*/\tau_x$ and $\sum_{j\in\cH}y_j^*/\tau_y$ on the weights of the nets (used in the previous section), one would arrive at the immediate conclusion that this would not lead to an approximation factor better than 2. Indeed, if $\sum_{i\in\cV}x_i^*=\sum_{j\in\cH}y_j^*:=X$ (e.g., this could be the case in any instance in which the sets $\cI_x$ and $\cI_y$ are isomorphic), then the weight of the obtained solution is $\frac{X}{\tau_x}+\frac{X}{\tau_y}\ge 4X=2z^*$, regardless of which  distribution we pick. To overcome this difficulty, we start by observing that the simple bound on the net, say $\frac{\sum_{i\in\cV}x_i^*}{\tau_x}$, can be tight only for a few values of $\tau_x$. So if we pick $\tau_x\in[0,1]$ randomly, we would expect the average size of an optimal net to be {\it strictly} smaller than $2\sum_{i\in\cV}x_i^*$. We prove that this is indeed the case if the $\cI_x$ is a set of unit intervals.

Let $V$ and $\cI$ be, respectively, a given set of points and a given set of intervals on the (say, horizontal) line. We say that $\cI$ satisfies the {\it continuity} property w.r.t. $V$ if each maximal clique in $\cI$ contains  a point from $V$. For $\tau\in[0,1]$, we denote by $\psi(\tau)=\psi(V,\cI,x,\tau)$ the size of an optimal $\tau$-net w.r.t. $(V,\cI,x)$, and let $\overline{\psi}(V,\cI,x):=\int_0^1\psi(\tau)\, d\tau$.

\begin{lemma}\label{l2}
Let $V$ be a given set of points on the (say, horizontal) line with non-negative values $x_v$, for $v\in V$, and $\cI$ be set of unit intervals on the line satisfying the continuity property w.r.t. $V$. Then $\overline\psi(V,\cI,x)\le \frac{119}{60} x(V)$.
\end{lemma}
Lemma~\ref{l2} follows from the following two claims.
\begin{claim}\label{cl1}
Let $V\subseteq[0,2]$ be a given set of points on the line with non-negative values $x_v$, for $v\in V$, and $\cI$ be set of unit intervals on $[0,2]$ satisfying the continuity property w.r.t. $V$. Then $\overline\psi(V,\cI,x)\le x(V)$.
\end{claim}
\begin{proof}
Since all intervals in $\cI$ are of length 1 and lie completely inside $[0,2]$, by the continuity property, there exists a point $v\in V$ such that $v\in I$ for all $I\in \cI$. It follows that 
$$
\psi(\tau)=\left\{
\begin{array}{ll}
1&\text{ if }\exists I\in\cI:~x(I)\ge\tau,\\
0&\text{ otherwise.}
\end{array}
\right.
$$
Thus, $\int_0^1\psi(\tau)\,d\tau=\max_{I\in\cI}x(I)\le x(V)$.
\end{proof}

\begin{claim}\label{cl2}
Let $V\subseteq[0,5]$ be a given set of points on the line with non-negative values $x_v$, for $v\in V$, and $\cI$ be set of unit intervals on $[0,5]$ satisfying the continuity property w.r.t. $V$. Then $\overline\psi(V,\cI,x)\le \frac{19}{12} x(V)$.
\end{claim}
\begin{proof}
By the continuity property, the size $\psi(\tau)$ of an optimal $\tau$-net cannot be more than $k:=4$. As we increase $\tau$ from $0$ to $1$, $\psi(\tau)$ decreases from $k$ to $0$. For $\ell\in[k]:=\{1,\ldots,k\}$, let $\tau_\ell$ be the smallest value of $\tau$ at which $\psi(\tau)$ drops below $\ell$ (note that some of the $\tau_\ell$'s might be equal); if $\psi(\tau)$ does not drop below $\ell$, set $\tau_\ell:=1$. Set further  $\tau_{k+1}:=0$. Then
\begin{align}\label{s2}
\int_0^1\psi(\tau)\, d\tau=\sum_{\ell=1}^k\int_{\tau_{\ell}}^{\tau_{\ell+1}}\psi(\tau)\, d\tau= \sum_{\ell=1}^k\ell(\tau_{\ell}-\tau_{\ell+1})=\sum_{\ell=1}^k\tau_{\ell}.
\end{align}
Let $\mI_\ell$ be the family of all independent sets (that is, subsets of pairwise-disjoint intervals) in $\cI$ of size $\ell$.
Suppose that $\psi(\tau)=\ell$. Then there exists an $\cI'\in\mI_\ell$  such that $x(I)\ge\tau$ for all $I\in\cI'$. This implies that $\tau_\ell\le\max_{\cI'\in\mI_\ell}\min_{I\in\cI'}x(I)$. It follows from \raf{s2} that
\begin{align}\label{s3}
\overline\psi(V,\cI,x)\le\sum_{\ell=1}^k\max_{\cI'\in\mI_\ell}\min_{I\in\cI'}x(I)=\max_{\cI_\ell\in\mI_\ell:\ell\in[k]}\sum_{\ell=1}^k\min_{I\in\cI_\ell}x(I).
\end{align}
By~\raf{s3}, it is enough to show that 
\begin{align}\label{s4}
\sum_{\ell=1}^k\min_{I\in\cI_\ell}x(I)\le \frac{19}{12}x(V), 
\end{align}
for any given families $\cI_\ell\in\mI_\ell$, for $\ell\in[k]$.
To show~\raf{s4}, we fix an arbitrary set of families $\cI_\ell\in\mI_\ell$, for $\ell\in[k]$, and consider the following LP and its dual:

{\centering \hspace*{-18pt}
	\begin{minipage}[t]{.47\textwidth}
		\begin{alignat}{3}
		\label{LP-P}
		\quad& \displaystyle \gamma^* = \max\quad \sum_{\ell=1}^k\alpha_\ell\\
		\text{s.t.}\quad & \displaystyle \sum_{i \in I}x_i\ge \alpha_\ell,\quad\forall I\in\cI_\ell,\quad \forall \ell\in[k],\label{LP-P-1}\\
		\quad&\sum_{i\in V}x_i=1,\label{LP-P-2}\\
		\quad&x_i\ge 0,\quad\forall i\in V,\label{LP-P-3}
		\end{alignat}
	\end{minipage}
	\,\,\, \rule[-25ex]{1pt}{25ex}
	\begin{minipage}[t]{0.47\textwidth}
		\begin{alignat}{3}
		\label{LP-D}
	\quad& \displaystyle \gamma^* = \min \quad\gamma \\
		\text{s.t.}\quad & \displaystyle \sum_{\ell=1}^k\sum_{I\in\cI_\ell:~i\in I}\beta_I\le \gamma,\quad\forall i\in V, \label{LP-D-1}\\
		\quad&\sum_{I\in\cI_\ell}\beta_I=1,\quad\forall\ell\in[k],\label{LP-D-2}\\
		\quad&\beta_I\ge 0,\quad\forall I\in\cI_\ell,\quad\forall\ell\in[k].\label{LP-D-3}
		\end{alignat}
\end{minipage}}

\medskip

\noindent It is enough to show that $\gamma^*\le\frac{19}{12}$; we do this by constructing a feasible dual solution with $\gamma=\frac{19}{12}$. Note that, for any given value of $\beta$ satisfying~\raf{LP-D-2} and~\raf{LP-D-3}, setting $\gamma:=\max_{i\in V}\sum_{I\in\cI_\ell:~i\in I}\beta_I$ gives a feasible dual solution.

Let $\cI_1:=\{I_1\}$, $\cI_2=\{I_2,I_3\}$, $\cI_3=\{I_4,I_5,I_6\}$, and $\cI_4=\{I_7,I_8,I_9,I_{10}\}$, where intervals in a given $\cI_\ell$ are ordered (from left to right) by their left endpoints. For simplicity, we write $\beta_j:=\beta_{I_j}$. We consider a number of cases:
\begin{itemize}
    \item {\bf Case 1}. There exists $I\in\cI_2$ such that $I\cap I_1=\emptyset$. W.l.o.g., $I=I_2$. In this case, we assign $\beta_1=\beta_2:=1$, $\beta_3:=0$, $\beta_4=\beta_5=\beta_6:=\frac13$, and $\beta_7=\beta_8=\beta_9=\beta_{10}:=\frac14$. Then $\gamma\le1+\frac13+\frac14=\frac{19}{12}$. 
 \end{itemize}   
 \noindent   Thus we may assume in the following cases that $I_1\cap I_2\ne\emptyset$ and $I_1\cap I_3\ne\emptyset$. Furthermore, there exists $I\in\cI_3$ such that $I\cap I_1=\emptyset$. W.l.o.g., $I=I_6$ (either $I=I_4$ or $I=I_6$; if $I=I_4$ we get a symmetric case).
\begin{itemize}
     \item {\bf Case 2}. There exists an interval $I\in\cI_4$ such that $I\cap I_1=I\cap I_6=\emptyset$. Then setting $\beta_1:=1$,  $\beta_2=\beta_3:=\frac12$, $\beta_4=\beta_5:=0$, $\beta_6:=1$, $\beta_I:=1$ and $\beta_{I'}:=0$ for all $I'\in\cI_4\setminus\{I\}$,  gives $\gamma\le 1+\frac12=\frac32$.
 \end{itemize}   
 \noindent Thus we may assume next that $I\cap I_1\ne\emptyset$ or $I\cap I_6\ne\emptyset$, for all $I\in\cI_4$. 
 The current set of assumptions imply that $I_7$ and $I_8$ both intersect $I_1$ but are disjoint from $I_6$, while $I_9$ and $I_{10}$ both intersect $I_6$ but are disjoint from $I_1$. These also imply that $I_2$ is disjoint from $I_6$, $I_9$ and $I_{10}$, $I_3$ is disjoint from $I_{10}$, $I_4$ is disjoint from $I_9$, and $I_4$ and $I_5$ are both disjoint from $I_{10}$ (see Figure~\ref{fig:f1}).
      \begin{figure}
    \centering
    \includegraphics[width=6in]{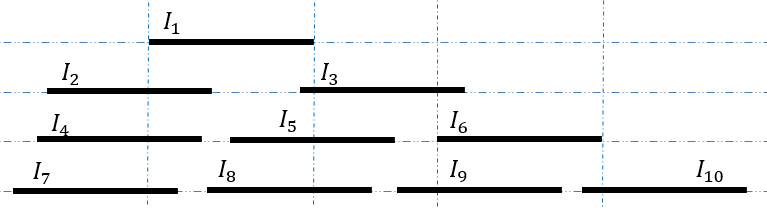}
    \caption{Illustration for the proof of Claim~\ref{cl2}.}
    \label{fig:f1}
\end{figure}
      Note that we may assume other overlaps between intervals that have not been forbidden above (e.g., whether $I_3$ overlaps with $I_6$ or not) are possible, as this would only add more constraints of the form ~\raf{LP-D-1} to the dual LP. To find all such overlaps, we construct a graph $G=([10],E)$ whose vertices are the indices of intervals $I_1,\ldots,I_{10}$, and whose edges are the set of pairs $\{j,k\}$ such that the overlap between $I_j$ and $I_k$ is forbidden. Then the maximal independent sets of $G$ will correspond one-to-one to the set of admissible overlaps. With our current assumptions,
      \begin{align*}
      E=&\big\{
       \{2,3\},\{4,5\},\{4,6\},\{5,6\},\{7,8\},\{7,9\},\{7,10\},\{8,9\},\{8,10\},\{9,10\},\\
       &\{1,6\},\{1,9\},\{1,10\},\{2,6\},\{2,9\},\{2,10\},\{3,10\},\{4,9\},\{4,10\},\{5,10\},\{6,7\},\{6,,8\}
      \big\}.
      \end{align*}
      The set of maximal independent sets in $G$ is
      \begin{align}\label{M}
      \cM=&\big\{\{1,2,4,7\},\{1,2,4,8\},\{1,2,5,7\},\{1,2,5,8\},
      \{1,3,4,7\}\{1,3,4,8\},\{1,3,5,7\},\{1,3,5,8\},\nonumber\\
      &\{3,5,9\},\{3,6,9\},\{6,10\}
      \big\}.
      \end{align}
      Each element $M\in\cM$ corresponds to a constraint of the form ~\raf{LP-D-1}. Thus, the dual-LP can be re-written in this case as follows:
      	\begin{align}\label{LP-DD}
	\quad& \displaystyle \gamma^* = \min \quad\gamma \\
		\text{s.t.}\quad & \displaystyle \sum_{j\in M}\beta_j\le \gamma,\quad\forall M\in\cM, \label{LP-DD-1}\\
		\quad&\beta_1=1,~\beta_2+\beta_3=1,~\beta_4+\beta_5+\beta_6=1,~\beta_7+\beta_8+\beta_9+\beta_{10}=1,\label{LP-DD-2}\\
		\quad&\beta_j\ge 0,\quad\forall j\in[10].\label{LP-DD-3}
		\end{align}
      
      Let us consider the assignment of $\beta$, satisfying~\raf{LP-DD-2} and~\raf{LP-DD-3}:
      $\beta_1:=1$, $\beta_2=\beta_3:=\frac12$, $\beta_4=\beta_5:=\frac1{12}$, $\beta_6:=\frac56$, $\beta_7=\beta_8:=0$, $\beta_9:=\frac14$ and $\beta_{10}:=\frac34$. Then,
      \begin{align*}
      \gamma^*\le \max_{M\in\cM}\sum_{j\in M}\beta_j=\max\{\beta_1+\beta_2+\beta_4+\beta_7,\beta_2+\beta_4+\beta_9,\beta_2+\beta_6+\beta_9,\beta_6+\beta_{10}\}=\max\Big\{\frac{19}{12},\frac{5}{6},\frac{19}{12},\frac{19}{12}\Big\}=\frac{19}{12}.
      \end{align*}

\end{proof}

Note that not all the maximal independent sets in $M$ given in \raf{M} are realizable together as maximal cliques in one configuration of intervals $\{I_1,\ldots,I_{10}\}$. For instance, if both $\{1,2,4,7\}$ and $\{1,2,4,8\}$ appear as maximal cliques in one configuration, then $I_4\cap I_8\neq\emptyset$, implying that $I_5\cap I_7=\emptyset$, and hence, $\{1,2,5,7\}$ and $\{1,3,5,7\}$ cannot be realized in the same configuration.  However, doing more case analysis based on such restrictions does not lead to any improvement of the bound obtained in Claim~\ref{cl2}.

\medskip
\noindent{\it Proof of Lemma~\ref{l2}}. Let $L$ be the smallest line segment containing all the intervals in $\cI$, and consider a $(5,\delta)$-shifted partition of $L$, where $\delta\in[0,5]$ is chosen randomly. We may  assume w.l.o.g. that $|L|$ is a multiple  of $5$. Let us denote by $L_1,\ldots,L_q$ the segments in the partition, and let $\cI_h=\{I\in\cI:~I\subseteq L_h\}$. Let $\cI':=\cI\setminus\bigcup_h\cI_h$ be the set of intervals crossed by the partition. Since each interval in $\cI'$ contains exactly one endpoint of some segment $L_i$, we can define a set $L_1',\ldots,L_{q}'$ of segments of length 2 such that the sets $\cI_h':=\{I\in\cI':~I\subseteq L_h'\}$, for $h=1,\ldots,k$, partition $\cI'$. Let $V_h=V\cap L_h$ and $V_h'=V\cap L_h'$, for $h=1,\ldots,q$. Then
\begin{align}\label{s1}
\overline\psi(V,\cI,x)\le\sum_{h=1}^q\overline{\psi}(V_h,\cI_h,x)+\sum_{h=1}^q\overline{\psi}(V_h',\cI_h',x),
\end{align}
as a $\tau$-net w.r.t. $(V, \cI,x)$ can be obtained by taking the union of two $\tau$-nets w.r.t. $(\cup_hV_h, \cup_h\cI_h, x)$ and $(\cup_hV_h', \cup_h\cI_h', x)$, respectively. By Claim~\ref{cl1}, $\overline{\psi}(V_h',\cI_h',x)\le x(V_h')$, for $h=1,\ldots,k$, and by Claim~\ref{cl2}, $\overline{\psi}(V_h,\cI_h,x)\le \frac{19}{12}x(V_h)$, for $h=1,\ldots,q$. Consider a point $v\in V$. Note that the probability that $v\in V_h'$ for some $h'\in\{1,\ldots,k\}$ is the same as the probability that a segment of length 2 centered at $v$ is crossed by the partition. By the random shifting property, the latter probability is $\frac25$. It follows that $\EE[\sum_{h=1}^kx(V_h')]\le \frac25x(V)$. Let us assume that we have fixed a shifted partition such that $\sum_{h=1}^kx(V_h')\le \frac25x(V)$. Putting 
the above together with \raf{s1} yields
\begin{align*}
\overline\psi(V,\cI,x)&\le\frac{19}{12}\sum_{h=1}^qx(V_h)+\sum_{h=1}^qx(V_h')\le\frac{19}{12}x(V)+\frac25x(V)=\frac{119}{60}x(V).
\end{align*}
The lemma follows.
\qed

The theorem below follows from Lemma~\ref{l2}.
\begin{theorem}\label{t2}
The integrality gap of unweighted continuous $\USS$ is at most $1.98\overline{3}$.
\end{theorem}
\begin{proof}
Let $(x^*,y^*)$ be an optimal solution for the LP relaxation~\raf{LP-RS}-\raf{LP-RS-2}, $V_x$ and $\cI_x$ be, respectively, the projections of $\cV$ and $\cR$ on the horizontal axis, and $V_y$ and $\cI_y$ be, respectively, the projections of $\cH$ and $\cR$ on the vertical axis. Clearly, we may think of $x^*$ and $y^*$ as values associated with the points in the sets $V_x$ and $V_y$, respectively. Pick $\tau_x\in[0,1]$ randomly, set $\tau_y=1-\tau_x$,
and let $S_x$ be an optimal $\tau_x$-net w.r.t. $(V_x,\cI_x,x^*)$ and $S_y$ be an optimal $\tau_y$-net w.r.t. $(V_y,\cI_y,y^*)$. These nets correspond to integral solutions $(\widehat x,\widehat y)$ to the given instance of \USS\ of expected value 
\begin{align*}
\EE[|S_x|]+\EE[|S_y|]&=\int_0^1 \psi(V_x,\cI_x,x^*,\tau)\,d\tau+\int_0^1 \psi(V_y,\cI_y,y^*,1-\tau)\,d\tau\\
&=\int_0^1 \psi(V_x,\cI_x,x^*,\tau)\,d\tau+\int_0^1 \psi(V_y,\cI_y,y^*,\tau)\,d\tau\\
&=\overline \psi(V_x,\cI_x,x^*)+\overline \psi(V_y,\cI_y,y^*).
\end{align*}
By Lemma~\ref{l2}, $\overline \psi(V_x,\cI_x,x^*)\le\frac{119}{60}\sum_{i\in\cV}x_i^*$ and $\overline \psi(V_y,\cI_y,y^*)\le\frac{119}{60}\sum_{j\in\cH}y_j^*$. The theorem follows.

\end{proof}

The proof of Lemma~\ref{l2} actually gives that $\overline\psi(V,\cI,x)\le\big(\alpha+\frac2k\big)$, whenever a bound of $\alpha\cdot x(V)$ can be proved for unit intervals defined on $[0,k]$. Thus, using the following version of Claim~\ref{cl2}, whose proof is given in the appendix, we get a bound of $1.9\overline3$ on the integrality gap of unweighted continuous $\USS$.

\begin{claim}\label{cl3}
Let $V\subseteq[0,6]$ be a given set of points on the line with non-negative values $x_v$, for $v\in V$, and $\cI$ be set of unit intervals on $[0,6]$ satisfying the continuity property w.r.t. $V$. Then $\overline\psi(V,\cI,x)\le \frac{8}{5} x(V)$.
\end{claim}

\paragraph{Limitation of the analysis.} Figure~\ref{fig:LB} shows a set of families $\cI_\ell$, $\ell=1,\ldots,k:=6$, of independent sets. We study this example for any $k\ge 4$. Let us assume that intervals in level $\ell$ are numbered from left to right as: $I_{\ell,1},\ldots,I_{\ell,\ell}$. 
Let 
\begin{align*}
\psi_1&:=\frac{1+\sqrt 5}{2},\qquad \psi_2:=\frac{1-\sqrt 5}{2},\\
c_1&:=\frac{7-4\psi_2}{\psi_1-\psi_2},\quad c_2:=\frac{4\psi_1-7}{\psi_1-\psi_2},\quad d_1:=\frac{1-\frac12\psi_2}{\psi_1-\psi_2},\quad d_2:=\frac{\frac12\psi_1-1}{\psi_1-\psi_2},\\
A_\ell&:=c_1\psi_1^\ell+c_2\psi_2^\ell-1, \quad B_\ell:=d_1\psi_1^\ell+d_2\psi_2^\ell-\frac12, \qquad\text{ for }\ell=0,1,2,\ldots,k-4,\\
\alpha_3(k)&:=\frac{B_{k-3}}{A_{k-3}},\qquad \gamma^*(k):=\alpha(k)+1.5,\\
\alpha_{\ell}(k)&:=A_{\ell-4}\alpha_3(k)-B_{\ell-4}\qquad\text{ for }\ell=0,1,2,\ldots,k.\\
\end{align*}
It can be verified that there is an optimal solution to the dual LP~\raf{LP-D}-\raf{LP-D-3} satisfying the following properties, where we write $\beta_{\ell,j}(k):=\beta_{I_{\ell,j}}$:
\begin{itemize}
    \item [(i)] level 1: $\beta_{1,1}(k)=1$; 
    \item[(ii)] level 2: $\beta_{2,1}(k)=\beta_{2,2}(k)=\frac12$;
    \item[(iii)] level 3: $\beta_{3,1}(k)=\beta_{3,2}(k)=\alpha_3(k)$, and $\beta_{3,3}(k):=1-2\alpha_3(k)$;
    \item[(iii)] level $\ell\in\{4,\ldots,k\}$: $\beta_{\ell,j}=0$ for $j=1,\ldots,\ell-2$, $\beta_{\ell,\ell-1}:=\alpha_\ell(k)$, and $\beta_{\ell,\ell}=1-\alpha_\ell(k)$.
\end{itemize}
\begin{figure}
    \centering
    \includegraphics[width=6in]{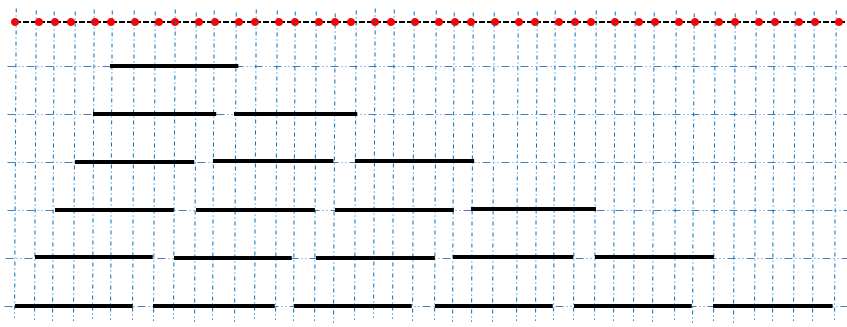}
    \caption{Lower bound example for the dual LP~\raf{LP-D}-\raf{LP-D-3}}
    \label{fig:LB}
\end{figure}
The corresponding optimal objective has value $\gamma^*:=\gamma^*(k)$. The first two values $\gamma^*(4)=\frac{19}{12}$ and $\gamma^*(5)=\frac54$ match the bounds given in Claims~\ref{cl2} and~\ref{cl3}, respectively. It can be easily verified that $\gamma^*(\infty)=\lim_{k\to\infty}\gamma^*(k)=\frac{35-\sqrt{5}}{20}\approx1.6382$. We conjecture that, for any $k\ge 1$ and any families of independent sets $\cI_1,\ldots,\cI_k$, the optimal value of the dual LP~\raf{LP-D}-\raf{LP-D-3} is at most $\gamma^*(\infty)$.\\

{\bf Remark.} Note that the above discussion does not imply any lower bound on the integrality gap of the standard relaxation of \USS\ or on the approximation factor of our rounding algorithm. Currently the best lower bound known on the intergrality gap is $3/2$ as shown by the following instance: $\cH$ and $\cV$  consist of two lines each, and for each pair of lines in $\cH \cup \cV$, there is a square intersecting exactly those two. Then, a fractional solution that assigns $0.5$ to each line in $\cH \cup \cV$ is a valid solution to the LP-relaxation with value $2$. On the other hand, the optimal solution must pick at least three of the lines. 

\section{Conclusions}
While we have been able to improve the upper bound on the integrality gap of two special cases of the rectangle stabbing problem, it seems clear that the upper bounds obtained are not tight. As we have pointed out, our analysis techniques cannot be used to obtain significantly better bounds. It seems that better and more general technique are necessary to make substantial progress on this. For the \USS\ problem, it is still open whether the standard LP relaxation of the weighted and/or discrete versions have integrality gap less than $2$. It is worth noting that if our conjecture in the previous section is true, then we would obtain an integrality gap of $\approx 1.6382$ for the standard LP-relaxations of both the discrete and continuous unweighted versions of \USS. Finally, obtaining an approximation factor better than 2 in polynomial time for the general rectangle stabbing problem remains a {\em very} interesting open question. 

%\bibliography{ref}
%\bibliographystyle{plainurl}

\appendix

\section{Proof of Claim~\ref{cl3}}
\begin{proof}
We proceed as in Claim~\ref{cl1}: we set $k=5$ and consider the dual LP~\raf{LP-D}-\raf{LP-D-3}. Again, it is enough to show that $\gamma^*\le\frac{8}{5}$; we do this by constructing a feasible dual solution with $\gamma=\frac{8}{5}$. Note that, for any given value of $\beta$ satisfying~\raf{LP-D-2} and~\raf{LP-D-3}, setting $\gamma:=\max_{i\in V}\sum_{I\in\cI_\ell:~i\in I}\beta_I$ gives a feasible dual solution.

Let $\cI_1:=\{I_1\}$, $\cI_2=\{I_2,I_3\}$, $\cI_3=\{I_4,I_5,I_6\}$, $\cI_4=\{I_7,I_8,I_9,I_{10}\}$, and $\cI_5:=\{I_{11},I_{12},I_{13},I_{14},I_{15}\}$, where intervals in a given $\cI_\ell$ are ordered by their left endpoints. For simplicity, we sometimes write $\beta_j:=\beta_{I_j}$. We consider a number of cases:
\begin{itemize}
    \item {\bf Case 1}. There exists $I\in\cI_2$ such that $I\cap I_1=\emptyset$. W.l.o.g., $I=I_2$. As $|\cI_5|=5$, there also exists $I\in\cI_5$ such that $I\cap I_1=\emptyset$ and $I\cap I_2=\emptyset$. In this case, we assign $\beta_1=\beta_2=\beta_I:=1$, $\beta_3:=0$, $\beta_4=\beta_5=\beta_6:=\frac13$, and $\beta_7=\beta_8=\beta_9=\beta_{10}:=\frac14$, and $\beta_{I'}=0$ for $I'\in\cI_5\setminus\{I\}$. Then $\gamma\le1+\frac13+\frac14=\frac{19}{12}<\frac85$. 
\end{itemize} 

    \noindent Thus we may assume in the following cases that $I_1\cap I_2\ne\emptyset$ and $I_1\cap I_3\ne\emptyset$. Furthermore, there exists $I\in\cI_3$ such that $I\cap I_1=\emptyset$. W.l.o.g., $I=I_6$ (either $I=I_4$ or $I=I_6$; if $I=I_4$ we get a symmetric case).
    
    In the following, we fix $\beta_1:=1$, $\beta_2=\beta_3:=\frac12$, $\beta_4=\beta_5:=\frac1{10}$, and $\beta_6:=\frac45$. We will argue that it is always possible to complete this assignment of the dual variables such that \raf{LP-D-2} holds and $\gamma\le\frac85$. 
    In fact, we will argue this can be done without assigning positive values to more than two variables at level $4$ (i.e., corresponding to intervals in $\cI_4$) and no more than three variables at level $5$. 
    We begin by assigning $\beta_I:=0$ for all $I\in\cI_4\cup\cI_5$ such that $I\cap I_1\ne\emptyset$. Note that at most {\it two} variables at level $4$ and at most two variables at level $5$ get assigned $0$ this way.
    
    For a given point $i\in V$, we denote by $$d(i):=\sum_{I~:~\beta_I\text{ has been assigned },~i\in I}\beta_I$$ the {\it current depth} of $i$. For an interval $I$, define further $d(I):=\max_{i\in I}d(i)$, and call it the current depth of $I$. We also define the {\it conservative} depth of $i\in V$, denoted by $\bar d(i)$, as an upper bound on the depth of $i$, implied by the current assignment (which takes into consideration the different possible locations of the intervals at levels 1, 2 and 3).
\begin{itemize}
    \item {\bf Case 2}. $I_3\cap I_6\ne\emptyset$. 
    The current configuration and partial assignment imply that 
    there exist at most one {\it unassigned} interval $I\in\cI_4$ and one {\it unassigned} interval $J\in\cI_5$ such that $I\cap I_3\cap I_6\neq\emptyset$, and $J\cap I_3\cap I_6\neq\emptyset$. (If there are two intervals, say in $\cI_4$, that overlap with $I_3\cap I_6$, then, by the assumptions that $I_1\cap I_3\neq\emptyset$ and $I_3\cap I_6\ne\emptyset$, the left-most one of these two intervals would overlap with $I_1$  and hence its corresponding dual variable would have already been set to $0$.) 
    Observe that, up to this point, we have at least {\it two} unassigned intervals $I,I'\in\cI_4$, and at least {\it three} unassigned intervals $J,J',J''\in\cI_5$. Among the intervals $I,I'$ (resp., $J,J',J"$), if there is one overlapping with $I_3\cap I_6$, we assume it is $I$ (resp., $J)$. It follows also from our current assumptions that $I_3$ does not overlap with any of the intervals $I'$, $J'$ and $J''$, and that either one of the latter two intervals, say $J'$, does not overlap with any of $I_1$, $I_3$ and $I_6$. Figure~\ref{fig:f2-1} shows the current configuration and partial assignment.
    \begin{figure}
        \centering
        \includegraphics[width=6in]{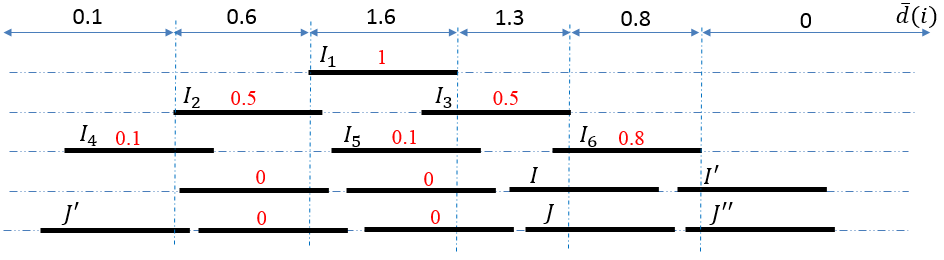}
        \caption{Partial assignment before considering Case 2 in Claim~\ref{cl2}. Numbers in red indicate the current partial assignment on the (variables corresponding to the) intervals. Numbers on the top indicate the (current) conservative depth $\bar d(i)$; for instance, if any of the assigned intervals, say $I_5$, is moved left or right, the current depth at any point remains within the conservative depth.}
        \label{fig:f2-1}
    \end{figure}
    We consider two further cases:
    \begin{itemize}
        \item {\bf Case 2.1}. $I\cap I_6\ne\emptyset$ (Figure~\ref{fig:f2-2}). Then, it is easy to verify that $d(I)\le 1.3$, $d(I')\le0.8$, $d(J)\le 1.3$, $d(J')\le 0.7$ and $d(J'')\le 1.5$. In this case, setting $\beta_I:=0.3$, $\beta_{I'}:=0.7$, $\beta_{J}:=0$, $\beta_{J'}:=0.9$, and $\beta_{J''}:=0.1$, would give $\gamma\le1.6$. 
        \begin{figure}
        \centering
        \includegraphics[width=6in]{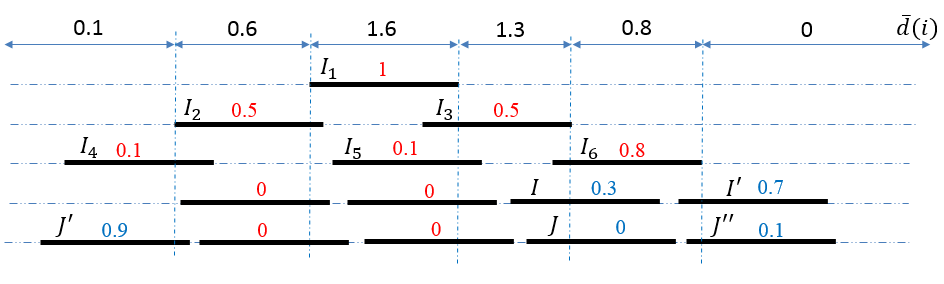}
        \caption{Assignment for Case 2.1 in Claim~\ref{cl2}. Numbers in red indicate the current partial assignment on the (variables corresponding to the) intervals. Numbers on the top indicate the (current) conservative depth $\bar d(i)$. Numbers in blue indicate the completed assignment.}
        \label{fig:f2-2}
    \end{figure}
    \item {\bf Case 2.2}. $I\cap I_6=\emptyset$ (Figure~\ref{fig:f2-3}). Then, it is easy to verify that $d(I)\le 0.6$, $d(I')\le0.8$, $d(J)\le 1.3$, $d(J')\le 0.6$ and $d(J'')\le 0.8$. In this case, setting $\beta_I:=0.5$, $\beta_{I'}:=0.5$, $\beta_{J}:=0.3$, $\beta_{J'}:=0.4$, and $\beta_{J''}:=0.3$, would give $\gamma\le1.6$. 
    \begin{figure}
        \centering
        \includegraphics[width=6in]{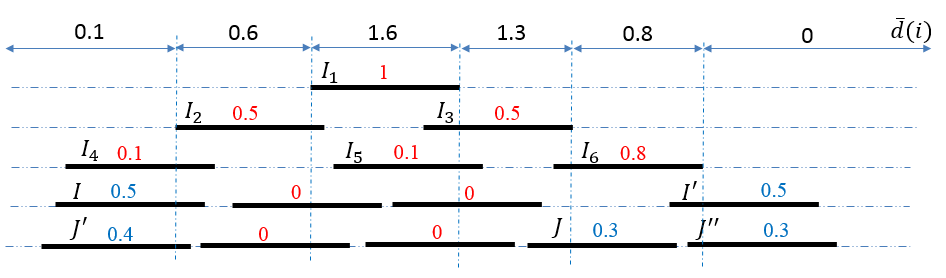}
        \caption{Assignment for Case 2.2 in Claim~\ref{cl2}.}
        \label{fig:f2-3}
    \end{figure}
    \end{itemize}
     \item {\bf Case 3}. $I_3\cap I_6=\emptyset$. Up to this point, we have at least {\it two} unassigned intervals $I,I'\in\cI_4$, and at least {\it three} unassigned intervals $J,J',J''\in\cI_5$. Among the latter three intervals, at least two, say $J',J''$ overlap with neither $I_1$ nor $I_3.$ We consider two cases:
     \begin{itemize}
         \item {\bf Case 3.1}. Either $J'$ or $J''$ does not overlap with $I_6$. W.l.o.g., $J'\cap I_6=\emptyset$ (Figure~\ref{fig:f3-1}). Then, it is easy to verify that $d(I)\le 0.8$, $d(I')\le0.8$, $d(J)\le 0.8$, $d(J')\le 0.8$ and $d(J'')\le 0.6$, and thus, setting $\beta_I=\beta_{I'}:=0.5$, $\beta_{J}=0.2$, $\beta_{J'}:=0.5$, and $\beta_{J''}:=0.3$, would give $\gamma\le1.6$. 
         \begin{figure}
        \centering
        \includegraphics[width=6in]{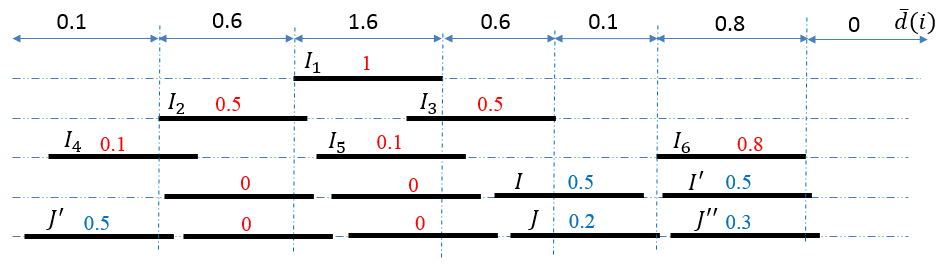}
        \caption{Assignment for Case 3.1 in Claim~\ref{cl2}.}
        \label{fig:f3-1}
        \end{figure}
         \item {\bf Case 3.2}. Both $J'$ and $J''$  overlap with $I_6$. Then $J$ does not overlap with $I_6$. Then, it is easy to verify that $d(I)\le 0.8$, $d(I')\le0.8$, $d(J)\le 0.6$, $d(J')\le 0.8$ and $d(J'')\le 0.8$, and thus, setting $\beta_I=0.4$, $\beta_{I'}:=0.6$, $\beta_{J}=0.6$, $\beta_{J'}=\beta_{J''}:=0.2$, would give $\gamma\le1.6$. 
    \begin{figure}
        \centering
        \includegraphics[width=6in]{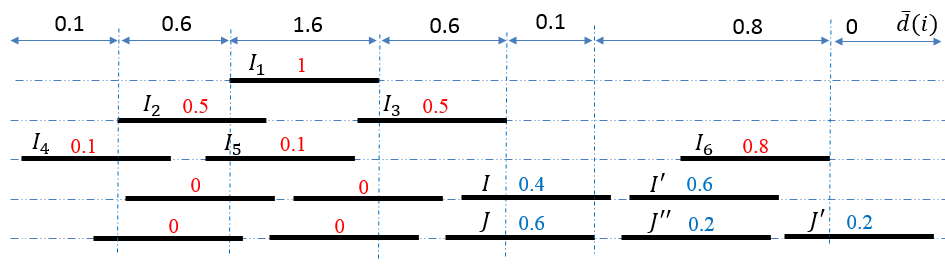}
        \caption{Assignment for Case 3.2 in Claim~\ref{cl2}.}
        \label{fig:f3-2}
    \end{figure}
    
     \end{itemize}

 \end{itemize}   
\end{proof}

\end{document}